\documentclass[letterpaper]{article}
\usepackage{aaai}
\usepackage{nicefrac}

\usepackage{times}
\usepackage{color}
\usepackage{latexsym}

\usepackage{graphicx}
\usepackage[centertags,fleqn]{amsmath}
\usepackage{amssymb,amsfonts,xspace}
\usepackage{booktabs}
\usepackage{theorem}

\newtheorem{theorem}{Theorem}

\newtheorem{corollary}{Corollary}

\newcommand{\qed}{\unskip\hspace*{1em}\hspace{\fill}$\Box$}
\newenvironment{proof}[1][Proof]{\begin{trivlist}
  \item[\hskip \labelsep {\it #1:}]}{%
    \qed\end{trivlist}}

\newcommand{\pav}[0]{\ensuremath{\mathit{PAV}}\xspace}

\newcommand{\av}[0]{\ensuremath{\mathit{AV}}\xspace}	
\newcommand{\sav}[0]{\ensuremath{\mathit{SAV}}\xspace}
\newcommand{\rav}[0]{\ensuremath{\mathit{RAV}}\xspace}

\newcommand{\pref}{\succsim\xspace}

\newcommand{\AB}{\ensuremath{{A}}}

\definecolor{officialblue}{RGB}{0, 93, 170}					

\begin{document}

\title{Computational Aspects of Multi-Winner Approval Voting}

\author{Haris Aziz \and Serge Gaspers\\ NICTA and UNSW \\ Sydney, Australia
\And
Joachim Gudmundsson \\ University of Sydney and NICTA\\ Sydney, Australia
\AND Simon Mackenzie, Nicholas Mattei \and Toby Walsh \\ NICTA and UNSW \\ Sydney, Australia}

\maketitle

\begin{abstract}
We study computational aspects of three prominent voting
rules that use approval ballots to elect multiple winners.
These rules are \emph{satisfaction approval voting},
\emph{proportional approval voting}, and
\emph{reweighted approval voting}.
We first show that
computing the winner for proportional
approval voting is NP-hard, closing a long standing open problem. As none of the rules
are strategyproof, even for
dichotomous preferences,
we study various strategic aspects of the rules. In particular,
we examine
the computational complexity of
computing a best response for both a single
agent and a group of agents. In many
settings, we show that it is NP-hard
for an agent or agents to compute how
best to vote given a fixed set of approval ballots from
the other agents.
\end{abstract}


\section{Introduction}

The aggregation of possibly conflicting preference is a central
problem in artificial intelligence \cite{Coni10a}. Agents
express preferences over candidates and a voting rule selects
a winner or winners based on these preferences.
We focus here on rules that select $k$ winners where $k$
is fixed in advance. This covers settings including
parliamentary elections, the hiring of
faculty members, and movie recommendation systems~\cite{OZE13a}.
Multi-winner rules can also be used to select a
committee~\cite{Ratl06a,LMM07a,ELS11a}.

Generally, in approval-based voting rules, an agent approves of (votes for)
a subset of the candidates.  The most straightforward way to aggregate these
votes is to have every approval for a candidate contribute one point to that candidate;
yielding the rule known as Approval Voting (\av).
Approval Voting is an obvious type of voting rule to
extend from the single winner to the multiple
winner case. Unlike, say, plurality voting
where agents nominate just their most preferred
candidate, approval ballots permit agents to
identify multiple candidates that they wish
to win.
Approval voting has many desirable
properties in the single winner case~\cite{Fish78d,BMS06a},
including its `simplicity,
propensity to elect Condorcet winners (when they exist),
its robustness to manipulation and its monotonicity'~\cite{LaSa10a}. However for the case of multiple 
winners, the merits of $\av$  are `less clear'~\cite{LaSa10a}. In particular,  
for the multi-winner case, $\av$ does address more egalitarian concerns such as proportional representation.


  Over the years, various methods for counting approvals
  have been introduced in the literature, each attempt to 
  address the fairness concerns 
  when using $\av$ for multiple winners~\cite{Kilg10a}.
One could, for instance, reduce the weight of an approval from a particular agent
based on how many other candidates the agent approves of have
been elected, as in \emph{Proportional Approval Voting (\pav)}.
Another way to ensure diversity across agents is vote
across a set of rounds.  In each round, the candidate with the most approvals
wins.  However, in each subsequent round we decrease the weight
of agents who have already had a candidate elected in earlier rounds;
this method is implemented in \emph{Reweighted Approval Voting (\rav)}.
Finally, \emph{Satisfaction Approval Voting (\sav)} modulates
the weight of approvals with a satisfaction
score for each agent, based on the ratio of approved candidates appearing
in the committee to the agent's total number of approved candidates.

These approaches to generalizing approval voting to the case
of multiple winners each have their own benefits and drawbacks.
Studying the positive or negative properties of these multi-winner
rules can help us make informed, objective decisions about which
generalization is better depending on the situations to which we are applying a particular multi-winner rule \cite{EFSS14a}. 
Though $\av$ is the most widely known of these rules, $\rav$ has been used, for example, in elections in Sweden. 
Rules other than $\av$ may have better axiomatic properties in the multi-winner setting and thus, motivate 
our study. For example, each of $\pav$, $\sav$, and $\rav$ have a more egalitarian objective than $\av$. 
Steven Brams, the main proponent of AV in single winner elections, has argued that $\sav$ is more suitable 
for equitable representation in multiple winner elections \cite{BrKi10a}. 

%
%
We undertake a detailed study of
computational aspects of $\sav$, $\pav$,
and  $\rav$.
We first consider the computational complexity
of computing the winner, a necessary result if any voting rule is expected to be used in 
practice. Although $\pav$ was introduced over
a decade ago, a standing open question has been the computational complexity of determining the winners, having 
only been referred to as ``computationally demanding'' before \cite{Kilg10a}.  We close this standing 
open problem, showing that winner determination for $\pav$ is NP-hard.
Our reduction applies to a host of
approval based, multi-winner rules in which
the scores contributed to an approved candidate by an agent
diminish as additional candidates approved by the agent are elected to the committee.

We then consider strategic voting for these rules.
We show that, even with
dichotomous preferences,
$\sav$, $\pav$ and $\rav$ are not strategyproof.
That is, it
may be beneficial for agents to mis-report their true
preferences. We therefore consider
computational aspects of manipulation.
We prove that finding the best response
given the preferences of other agents is NP-hard
under a number of conditions for $\pav$, $\rav$, and $\sav$.
In particular, we examine the complexity
of checking whether  an
agent or a set of agents can make a given candidate or
a set of candidates win.
These results offer support for $\rav$ over $\pav$ or $\sav$ as it is the only rule for which winner determination
is computationally easy but manipulation is hard.

\section{Related Work}

An important branch of social choice concerns determining
how and when agents can benefit by misreporting
their preferences. In computational social choice, this problem
is often studied through the lens of computational complexity 
\cite{BTT89a,FaPr10a,FHH10a}.  If it is computationally hard
for an agent to compute a beneficial misreporting of their preferences for a particular
voting rule, the rule is said to be resistant to manipulation.  If it 
is computationally difficult to compute a misreport, agents may decide to be truthful, 
since they cannot always easily manipulate. 
Connections have been made between manipulation and
other important questions in social choice 
such as deciding when to terminate preference elicitation and 
determining possible winners \cite{KoLa05a}.

Surprisingly, there has only been limited consideration
of computational aspects of multi-winner elections.
Exceptions include work by Meir et al.~\shortcite{MPR08b} which
considers single non-transferable voting, approval voting,
$k$-approval, cumulative voting
and the proportional schemes of Monroe, and of Chamberlin and
Courant. Most relevant to our study is that
for approval voting, Meir et al.
prove that manipulation with general
utilities and control by adding/deleting
candidates are both polynomial to compute, but
control by adding/deleting agents
is NP-hard. Another work that considers
computational aspects of multi-winner elections is \citeauthor{OZE13a}~\shortcite{OZE13a},
but their study is limited to $k$-approval and scoring
rules. Finally, the control and bribery problems for $\av$ and two other approval voting variants are well catalogued
by \citeauthor{BEH+10a}~\shortcite{BEH+10a}, demonstrating that $\av$ is generally resistant to bribery but 
susceptible to most forms of control when voters have dichotomous utility functions.

The Handbook of Approval Voting discusses
various approval-based multi-winner rules
including $\sav$, $\pav$ and $\rav$.
Another prominent multi-winner rule in the Handbook
is \textit{minimax approval voting}~\cite{BKS07a}.
Each agent's approval ballot and the winning set can be seen as a binary vector.
Minimax approval voting selects
the set of $k$ candidates that minimizes the maximum
Hamming distance from the submitted ballots.
Although minimax approval voting is a natural
and elegant rule, LeGrand et al.~\shortcite{LMM07a} showed that computing
the winner set is unfortunately NP-hard.
Strategic issues and approximation questions
for minimax approval voting are covered in \cite{CKM10a}
and \cite{GNR13} where the problem is known
as the ``closest string problem.''


The area of multi-winner approval voting is closely related to
the study of proportional representation when
selecting a committee (Skowron et al.~\citeyear{SFS13a,SFS13b}).
Ideas from committee selection have therefore been used in computational
social choice to ensure diversity when selecting a collection
of objects \cite{LuBo11a}.  Understanding approval voting schemes which 
select multiple winners, as the rules we consider often do, is an important
area in social choice with applications in a variety of settings from committee
selection to multi-product recommendation \cite{EFSS14a}.

\section{Formal Background}

We consider the social choice setting $(N,C)$ where $N=\{1,\ldots, n\}$ is the
set of agents and $C=\{c_1,\ldots, c_m\}$ is the set of candidates. Each
agent $i\in N$ has a complete and transitive preference relation $\pref_i$
over $C$. Based on these preferences, each agent expresses an approval
ballot $A_i\subset C$ that represents the subset of candidates that he
approves of, yielding a set of approval ballots $\AB = \{A_1,\ldots, A_n\}$. We will consider approval-based multi-winner rules that take
as input $(C,\AB, k)$ and return the subset $W \subseteq C$ of size $k$ that is the winning set.

\subsection{Approval Voting ($\av$)}
\av finds a set $W \subseteq C$ of size $k$ that maximizes the total score $App(W)=\sum_{i\in N}{|W \cap A_i|}$.
That is, the set of \av winners are those candidates receiving the most points across all submitted ballots.
\av has been adopted by several academic and
professional societies such as the American Mathematical Society (AMS),
the Institute of Electrical and Electronics Engineers (IEEE), and
the International Joint Conference on Artificial Intelligence.

\subsection{Satisfaction Approval Voting ($\sav$)}
An agent's satisfaction is the fraction of his or her approved
candidates that are elected. \sav maximizes the sum of such
scores. Formally,
\sav  finds a set $W \subseteq C$ of size $k$ that maximizes
$Sat(W)=\sum_{i\in N}\frac{|W\cap A_i|}{|A_i|}$.
The rule was proposed by~\cite{BrKi10a} with the aim
of representing more diverse interests than~$\av$.

\subsection{Proportional Approval Voting ($\pav$)}

In \pav, an agent's satisfaction score is $1+1/2+1/3 \cdots 1/j$
where $j$ is the number of his or her
approved candidates that are selected in $W$.
Formally, $\pav$ finds a set $W \subseteq C$ of size $k$
that maximizes the total score $\pav(W)=\sum_{i\in N}r(|W\cap A_i|)$ where $r(p)=\sum_{j=1}^p\frac{1}{j}$.
$\pav$ was proposed by the mathematician Forest Simmons in 2001
and captures the idea of diminishing returns --- an individual agent's preferences
should count less the more he is satisfied.

\subsection{Reweighted Approval Voting ($\rav$)}

\rav converts \av into a multi-round rule, selecting a candidate
in each round and then reweighing the approvals for the subsequent rounds.
In each of the $k$ rounds, we select an unelected candidate to add to
the winning set $W$ with the highest ``weight'' of approvals. In each round we reweight each agents
approvals, assigning for all $i \in N$ the weight $\frac{1}{1+|W \cap A_i|}$ to agent $i$.
$\rav$ was invented by the
Danish polymath Thorvald Thiele in the early 1900's. $\rav$ has also been referred to as ``\emph{sequential proportional AV}''~\cite{BrKi10a},
 and was used briefly in Sweden during the early 1900's.

\smallskip	
Tie-breaking is an important issue to consider when 
investigating the complexity of manipulation and winner determination
problems as it can have a significant impact on the complexity 
of reasoning tasks (Obraztsova et al.~\citeyear{OEH11a}, \citeauthor{ObEl11a}~\citeyear{ObEl11a}, Aziz et al.~\citeyear{AGM+13a}).  We make the worst-case 
assume that a tie-breaking rule takes the form of a linear 
order over the candidates that is given as part of the problem input and favors
the preferred candidate; as is common 
in the literature on manipulation (Bartholdi et al.~\citeyear{BTT89a}, \citeauthor{FaPr10a} \citeyear{FaPr10a}, Faliszewski et al.~\citeyear{FHH10a}).  Note that many of our proofs
are independent of the tie-breaking rule, 
in which case the hardness results transfer to any arbitrary tie-breaking rule.

\section{Winner Determination}
	
We first examine one of the most basic computational questions,
computing the winners of a voting rule.

\vspace{.1cm}
\noindent
\textbf{Name}: \textsc{Winner Determination (WD)}. \\
\noindent
\textbf{Input:} An approval-based voting rule $R$,
a set of approval ballots $\AB$ over the
set $C$ of candidates, and a committee size $k\in \mathbb{N}$.\\
\noindent
\textbf{Question:} What is the winning set, $W\subseteq C$, with $|W|=k$?
\vspace{.1cm}

Firstly, we observe that WD is polynomial-time computable for $\sav$, $\rav$, and $\av$. Although $\rav$ is polynomial-time to compute, it
has been termed ``computationally difficult'' to analyze
in \cite{Kilg10a}. We
provide support for this claim
by showing that computing a best response for $\rav$ is NP-hard (Theorem~\ref{thm:RAV-respons}).  We close the computational complexity of WD for $\pav$ in this section.


\begin{theorem} \label{thm:pav-coNP}
 WD for $\pav$ is NP-complete, even if each agent approves of two candidates.
\end{theorem}
\begin{proof}
The problem is in NP since we merely need as
witness a set of candidates with $\pav$ score $s$.

To show hardness we give a reduction from the NP-hard \textsc{Independent Set} problem \cite{GaJo79}:
Given $(G,t)$, where $G = (V, E)$ is
an arbitrary graph and $t$ an integer, is there an independent set of size $t$ in $G$.  An independent
set is a subset of vertices $S \subseteq V$ such that no edge of $G$ has both endpoints in $S$.
For a graph $G$, we build a $\pav$ instance for which
a winning committee of size $t$ corresponds to
an independent set in $G$ of size $t$, and vice-versa.

Consider a graph $G=(V,E)$, and define the following $\pav$ instance, $(N,C,\AB, k)$:
We have a set of agents $N$ and a set of candidates $C$. For each vertex $v \in V$,
we create $deg(G)-deg(v)$ `dummy' candidates in $C$, where $deg(G)$ is the maximum degree of
$G$, $deg(G) > 1$, and $deg(v)$ the degree of vertex $v$. For each $v \in V$, we also create
another candidate in $C$, labeled $C_v$. We create an agent in $N$ for each
edge $e \in E$. For each vertex we also create $deg(G)-deg(v)$ agents. Each of the edge
agents approves of the two candidates corresponding to the
vertices connected by the edge. Each vertex agent associated with vertex $v$
approves of $C_v$ and one of the dummy candidates associated with $v$, thus
each dummy candidate has exactly one agent approving of him. We also set $k = t$.

We will show that there is a committee of size $k = t$ scoring a total approval of at least $s=deg(G) \cdot t$
if and only if $G$ has an independent set of size $t$.
First, note that adding a candidate to a committee increases the total score of the committee by
at most $deg(G)$, since at most $deg(G)$ agents see their satisfaction score rise by at most one.
Also, if adding a candidate $c$ to a committee increases the total score of the committee by
exactly $deg(G)$, then $c$ corresponds to a vertex in $G$, since each dummy vertex is approved by only one agent,
and the vertex corresponding to $c$ is not adjacent to a vertex corresponding to any other candidate in
the committee.
Thus, the candidates in a committee of size $k = t$ scoring a total approval of $s$
correspond to an independent set of size $t$ in $G$ and vice-versa.
\end{proof}

The reduction in this proof actually implies a stronger result,
namely that, unless FPT=W[1], WD for $\pav$
cannot be solved in time $f(k)\cdot m^{O(1)}$, for any function $f$,
even if each agent approves of two candidates.
This is because it is a parameterized reduction where the parameter $k$ is a function
of the parameter $t$ for \textsc{Independent Set}, which is W[1]-hard for parameter $t$~\cite{DoFe13a}.
Thus, even for relatively small committee sizes, a factor $m^k$ in the running time seems unavoidable.

\begin{corollary}
 WD for $\pav$ is W[1]-hard.
\end{corollary}

\section{Strategic Voting}

As in the single winner case,
agents may benefit from mis-reporting
their true preferences when electing multiple winners.
We consider the special case of
dichotomous preferences where
each agent has utility 0 or 1 for
electing each particular candidate.
In this case, we say that
a multi-winner approval-based
voting rule is \emph{strategyproof} 
if and only if there does not exist an agent
who has an incentive to approve a candidate
with zero utility and does not have an incentive
to disapprove a candidate for whom the agent
has utility 1. We note that for \emph{dichotomous preferences}, $\av$ is strategyproof (if lexicographic tie-breaking is used).
However, it is polynomial-time manipulable for settings with more general utilities \cite{MPR08b}.
On the other hand, $\sav$, $\pav$ and $\rav$ are not.

\renewcommand{\theenumi}{(\roman{enumi})}

\begin{theorem}
 	$\sav$, $\pav$ and $\rav$ are not
strategyproof with dichotomous preferences.
\end{theorem}
\begin{proof}
	We treat each case separately. We assume that ties are always broken
	lexicographically with $a \succ b \succ c$, e.g., $\{a,b\}$ is preferred to $\{a,c\}$.
	\begin{enumerate}
		
\item For $\sav$, assume $k=2$, $C = \{a, b, c\}$, and agent $1$ has non-zero utility only for $a$ and $b$. 	 Let,
\[	A_2=\{a\}, A_3=	\{a\}, A_4=\{a\}, A_5=\{c\}, A_6=\{b,c\}.\]
The outcome is $\{a,c\}$ if $A_1=\{a,b\}$, but if agent $1$ only approves $b$, the outcome is $\{a,b\}$ which has the maximum utility and is preferred by tie-breaking.


\item For $\pav$, consider the same setting but now with the following votes:
\[A_2=\{b\}, A_3=\{a,c\}, A_4=\{a,c\}, A_5=\{c\}.\]
      The outcome $\{a,b\}$ is only possible if agent 1 approves only $b$. Otherwise it is $\{a,c\}$.

\item 	For $\rav$, consider the same setting but now with the following votes:
\[A_2=\{a\}, A_3=\{a\}, A_4=\{a\}, A_5=\{c\}, A_6=\{b,c\}.\]
	 	The outcome is $\{a,c\}$ for all reported preferences of agent $1$
		 $A_1=\{b\}$, in which case the outcome is $W = \{a,b\}$.
\end{enumerate}
This completes the proof.
\end{proof}

%
%

With \sav, \pav and \rav, it can therefore be beneficial for agents to
vote strategically.
Next, we consider the computational complexity of computing such strategic votes.

%


\vspace{.2cm}
\noindent
\textbf{Name:} \textsc{Winner Manipulation (WM)} \\
\noindent
\textbf{Input:}
An approval-based voting rule $R$,
a set of approval ballots $\AB$ over the
set $C$ of candidates, a winning set size $k$, a number of agents
$j$ still to vote, and a preferred candidate~$p$.\\
\noindent
{\bf Question:} Are there
$j$ additional approval ballots
so that $p$ is in the winning set $W$ under $R$?

\vspace{.2cm}
\noindent
{\bf Name:} {\sc Winning Set Manipulation (WSM)}. \\
\noindent
{\bf Input:}
An approval-based voting rule $R$,
a set of approval ballots $\AB$ over the
set $C$ of candidates, a winning set size $k$, a number of agents
$j$ still to vote, and a set of preferred candidates $P \subseteq C$.\\
\noindent
{\bf Question:} Are there
$j$ additional approval ballots
such that $P$ is the winning set of candidates
under $R$?
\vspace{.2cm}

We note if \textsc{WM} or {\sc WSM} is NP-hard for a single agent ($j=1$),
then the more general problem of maximizing the utility of an agent is also NP-hard.
For $\av$, the utility maximizing best response of a single agent can be computed in polynomial time~\cite{MPRZ08a}.
We note our definitions have additive utilities, and the question is to 
cast $j$ votes so as to maximize the \emph{total} utility. 
This is more general than WM/WSM, 
since a simple reduction from gives utility 1 to the candidates in 
P (or \{p\}), and 0 to all the other candidates.

%
%
%

\subsection{Satisfaction Approval Voting ($\boldsymbol{\sav}$)}
	
\textsc{WM} under \sav is polynomial-time solvable.
The agents cast an approval ballot for just the preferred candidate.
This is the best that they can do. If the preferred candidate does not win
in this situation, then the preferred candidate can never win. It follows that we can also
construct the set of candidates that can possibly win in polynomial time.
It is more difficult to decide if a given $k$-set of candidates can
possibly win. With certain  voting rules, this problem simplifies if the optimal strategy of $j$ manipulating agents need to cast only one form of vote.
This is not the case with $\sav$.
	
\begin{theorem}
    To ensure a given set of candidates is selected under $\sav$, the manipulating coalition may need to cast a set of votes that are not all identical.
\end{theorem}

\begin{proof}
Suppose $k=3$ and $C=\{a, b, c, d, e, f, g\}$, one agent approves both $a$ and $b$,
and three agents approve $d$, $e$, $f$ and $g$. If there are two more agents
who want $a$, $b$ and $c$ to be elected, then one agent needs to approve
$c$ and the other both $a$ and $b$, or one agent needs to
approve $a$ and $b$, and the other $a$ and $c$.
\end{proof}

This makes it difficult to decide how a coalition of agents must vote.
In fact, it is intractable in general to decide if a given
set of candidates can be made winners. We omit the proof 
for space but observe that it is a reduction to the 
permutation sum problem as in the NP-hardness 
proof for Borda manipulation with two agents \cite{DKNW11a}.

\begin{theorem}\label{thm:wsm-nphard-sav}
\textsc{WSM} is NP-hard for $\sav$.
\end{theorem}

The proof requires both the number of agents and
the size of the winning set to grow. An open question
is the computational complexity when we bound
either or both the number of agents and the
size of the winning set.
We can also show that
it is intractable to manipulate
$\sav$ destructively. 

We can adapt the proof for Theorem~\ref{thm:wsm-nphard-sav} to show the following
statement as well.

\begin{theorem}\label{thm:sav-destructive}
	For $\sav$, it is not possible for a single manipulator to compute in polynomial time a vote
	that maximizes his utility, unless P=NP.
\end{theorem}


Hence, in the case of multi-winner
voting rules, destructive manipulation
can be computationally harder than
constructive manipulation. This contrasts to
the single winner case where destructive
manipulation is often easier than constructive
manipulation \cite{CSL07a}.
It also follows from Theorem~\ref{thm:sav-destructive} that
it is intractable to manipulate
$\sav$
to ensure a given utility or greater.


We next turn to the special cases of a single agent and a pair of agents. Winning set
manipulation is polynomial with either one or two agents left to vote.
This result holds even if the size
of the winning set is not bounded (e.g. 
$k=m/2$).  The proofs are one agent is omitted for space, however
we observe that the proof of the following Theorem can be extended for 
the case where a set $P$ has to be a subset of the winning set.

\begin{theorem} \label{thm:sav-two agents}
If two agents remain to vote,
\textsc{WSM} is polynomial for $\sav$.
\end{theorem}

\subsection{Proportional Approval Voting ($\boldsymbol{\pav}$)}

The proof of the NP-hardness of \textsc{Winner Determination} for \pav can be adapted
to also show that basic manipulation problems are coNP-hard for $\pav$.


\begin{theorem}\label{thm:wm-wsm-conp}
 For $\pav$, \textsc{WM} and \textsc{WSM} are coNP-hard, even if there is no manipulator.
\end{theorem}
%

In Theorem~\ref{thm:wm-wsm-conp} the hardness of WM and WSM really comes from the hardness of WD, demonstrated by requiring no manipulators.
This result motivates us to investigate the situation where a ``real'' manipulation is necessary, that is, 
whether a single manipulator can include a particular candidate in the winning set, even 
if WD is polynomial-time computable for the underlying $\pav$ instance.  While we conjecture this is hard,
we can formally prove the following, slightly weaker, statement.

\begin{theorem}
For $\pav$, it is not possible for a single manipulator to compute in polynomial time a vote
that maximizes his utility, unless $P=NP$.
\end{theorem}

\subsection{Reweighted Approval Voting ($\boldsymbol{\rav}$)}
In \rav the decision for a single agent of whom to vote for in order to maximize his
utility is not straightforward.  Suppose we are selecting a committee of size $k=2$ with
$C=\{a,b,c,d\}$:
\begin{align*}
&A_2=\{b,d\},  A_3= \{c,d\}, A_4= \{a,b,c,d\}\\
&A_5=A_6=\{b,c,d\},
A_7= \{a,b\},
A_8= \{c\},
A_9= \{a\}.
\end{align*}

	

If the agent wants to elect $a$ to the committee then he may need to express preference for \textit{more} than just his choice set.
In the above example, if agent $1$, casts the ballot $A_1=\{a\}$ then in Round 1
$b$ is elected, in Round 2 $c$ is elected.
However, if the agent casts the ballot $\{a, d\}$ then in Round 1 $d$ is elected, 
and in Round 2 $a$ is elected.




\begin{theorem}\label{single:over}
    Under $\rav$, an agent who wants to include a single candidate in the committee may have incentive to approve more candidates than $P$.
\end{theorem}

Furthermore, if the agent is attempting to fill a committee with a preferred set of candidates,
he may have incentive not to approve some candidates so that they may be elected.
Suppose we are selecting a committee of size $k=3$ with
$C=\{a,b,c,d\}$, using lexicographic tie-breaking:
\begin{align*}
&A_2=\{b,d\},  A_3= \{c,d\}, A_4=A_5=A_6= \{b,c,d\}\\
&A_7= \{b\},
A_8= \{c\},
A_9=A_{10}= \{a\}.
\end{align*}
If the
agent has favored set $\{a, b, d\}$ and he approves all of them, then in Round
1 $b$ is elected, in Round 2 $c$ is elected, and in Round 3 $a$ is elected.
%
However, if the agent casts the ballot $\{a, d\}$ then in Round 1 $d$ is elected,
in Round 2 $a$ is elected, and in Round 3 $b$ is elected, exactly the favored set.
%


If a manipulator wants to elect exactly a favored set $P$ then he must approve
either $P$, or a subset of it.  
%

\begin{theorem}
    Under $\rav$, an agent who wants to elect an exact set of candidates will never have an incentive to approve
     a superset of his preferred candidates, though he may have an incentive to approve a subset of them.
\end{theorem}
	
\begin{theorem} \label{thm:RAV-respons}
    For $\rav$, \textsc{WM} is NP-hard.
\end{theorem}

\begin{proof}
To show that $\rav$ is NP-hard to manipulate we reduce from 3SAT.
Given a instance of 3SAT with $w$ variables $\Phi = \{\phi_1, \ldots, \phi_w\}$,
$t$ clauses $\Psi = \{\psi_1, \ldots, \psi_t\}$, inducing $2w$ literals
$\{l_1, \ldots l_{2w}\}$. We construct an instance of $\rav$, $(C,\AB, k)$ 
where a manipulator's preferred candidate $p$ is in the winning set if and only if 
there is an assignment to the variables in $\Phi$ such that all clauses are satisfied.

For each variable $\phi_i$
introduce 2 candidates in $C$, corresponding to the positive
and negative literal of that variable, and $2n-i$ agents approving of the
2 candidates; note that $n \gg w+t $.
For each clause $\psi_j$ introduce two additional new
candidates, corresponding to the clause being satisfied or unsatisfied,
along with $2n - w - j$ new agents approving of both the two
new candidates.
Additionally for each clause $\psi_j$, we add an agent in $\AB$ approving
each of the candidates that correspond to the positive and
negative literals in $\psi_j$; this ensures that both the positive
and negative literal have the same weight of approval in the
set of agents.
We also need to add 2 agents approving of the candidate
corresponding to the negation of the clause to maintain the weighting.
Finally, add an extra 2
candidates to $C$, $a$ and $b$. We add 2 agents approving of the candidate corresponding to
a clause being unsatisfied, and 2 agents approving of the the candidate corresponding to
each clause being satisfied and approving of $b$

Add $t$ agents approving of $a$.
The size of the winning set $k$ is equal to $|\Phi| + |\Psi| + 1$.  Intuitively,
the manipulator must approve of a setting of all the variables in the original
3SAT instance that satisfies all the clauses, plus the preferred candidate.
We can now see that the manipulating agent is
only capable of ensuring candidate a is elected by computing a solution
to the initial 3SAT instance.
\end{proof}

The above proof also shows it is NP-hard to determine if $P$ can be 
made a subset of the winning set, $P \subseteq W$.






 \begin{table}[t]
	\centering
\begin{tabular}{lccc}
\toprule
&WD&WM&WSM\\
\midrule
$\av$&in P&in P&in P\\
$\sav$&in P&in P&NP-h\\
$\pav$&NP-h&coNP-h&coNP-h\\
$\rav$&in P&NP-h&-\\
\bottomrule
\end{tabular}
\caption{Summary of computational results for approval-based multi-winner rules for \textbf{W}inner \textbf{D}etermination, \textbf{W}inner \textbf{M}anipulation, and \textbf{W}inning \textbf{S}et \textbf{M}anipulation.}
 \label{table:summary:av}
 \end{table}

\section{Conclusions}

We have studied some basic computational questions
regarding three prominent voting
rules that use approval ballots to elect multiple winners.
We closed the computational complexity of
computing the winner for \pav and studied
the computational complexity of
computing a best response for a variety
of approval voting rules. In many
settings, we proved that it is NP-hard
for an agent or agents to compute how
best to vote given the other approval ballots.
 To complement this
complexity study, it would be interesting
to undertake further axiomatic and empirical analyses of \pav, \rav, and
\sav. Such an analysis would
provide further insight into the
relative merits of these rules.

\subsubsection*{Acknowledgements}
NICTA is funded by the Australian Government through the Department of Communications and the Australian Research Council through the ICT Centre of Excellence Program. 
Serge Gaspers is the recipient of an Australian Research Council Discovery Early Career Researcher Award (project number DE120101761). Joachim Gudmundsson is funded by the Australian Research Council (project number FT100100755).

\bibliographystyle{aaai}
\bibliography{abbshort,adt}
\end{document}